\documentclass[letterpaper, 10 pt, conference]{ieeeconf}  

\IEEEoverridecommandlockouts                              

\usepackage{amsmath} 
\usepackage{amssymb}  
\usepackage{graphicx}

\usepackage{algorithm}
\usepackage{algpseudocode}

\usepackage{amsthm}
\usepackage{hyperref}
\usepackage{xcolor}
\usepackage[noadjust]{cite}
\usepackage{array}

\usepackage[english]{babel}
\newtheorem{theorem}{Theorem}

\newtheorem{lemma}[theorem]{Lemma}
\newtheorem{remark}{Remark}
\newtheorem{proposition}{Proposition}
\newtheorem{assumption}{Assumption}

\newcommand{\Pt}{\widetilde{P}}
\newcommand{\Ph}{\widehat{P}}

\newcommand{\Eh}{\widehat{E}}
\newcommand{\Fh}{\widehat{F}}
\newcommand{\Gh}{\widehat{G}}
\newcommand{\Eb}{\bar{E}}
\newcommand{\Fb}{\bar{F}}
\newcommand{\Gb}{\bar{G}}
\newcommand{\T}{^{\!\top}}

\title{\LARGE \bf
Shifted-Chart Composition for Horizon-Optimal Time-Varying LQR
}

\author{Viswas Dinesh and Spandan Roy
\thanks{Viswas Dinesh and Spandan Roy are with Robotics Research Centre,
        IIIT Hyderabad, India. Emails:
        {\tt\small viswas.dinesh@research.iiit.ac.in, spandan.roy@iiit.ac.in}}
}

\makeatletter
\def\ps@IEEEtitlepagestyle{%
    \def\@oddhead{%
        \hfil{\scriptsize\thepage}%
    }%
    \let\@evenhead\@oddhead
    \def\@oddfoot{%
        \parbox[t]{\textwidth}{\centering\scriptsize
            Submitted to the 2027 American Control Conference (ACC); under review.\\
            This work has been submitted to the IEEE for possible publication.
            Copyright may be transferred without notice, after which this version
            may no longer be accessible.}%
    }%
    \let\@evenfoot\@oddfoot
}

\makeatother

\begin{document}

\maketitle
\thispagestyle{IEEEtitlepagestyle}
\pagestyle{empty}

\begin{abstract}
Horizon optimization in finite-horizon LQR treats the horizon $T$ as a
decision variable and adds a penalty $wT$ to the cost. Recent
composition-based methods evaluate all candidate horizons in $O(Nn^3)$,
instead of the $O(N^2n^3)$ of one Riccati recursion per horizon, but they
track the inverse cost-to-go $P^{-1}$ and therefore require every stage and
terminal cost to be positive definite, whereas LQR requires only positive
semidefiniteness. This excludes costs that leave some states unpenalized. In this work, we
track instead the shifted inverse $\widehat{P}_k=(P_k+cI)^{-1}$, the inverse
coordinate after a symplectic shear of the state--costate variables. Its
domain contains every positive semidefinite Riccati matrix, and
$P_k=\widehat{P}_k^{-1}-cI$ recovers the Riccati solution. We show that the
shift can be absorbed into the stage coefficients of the inverse-coordinate
composition, so the $O(Nn^3)$ search carries over, exact in exact
arithmetic, subject to invertibility conditions. These hold for all small
$c$ when the dynamics contract the unpenalized directions, fail for at most
finitely many $c$ outside a degenerate case, and are certified per horizon
at run time. In our experiments, diagonal loading $Q_k+\varepsilon I$ either
left cost errors several orders of magnitude above those of the shifted
chart or, for small fixed $\varepsilon$, selected the wrong horizon. The
shifted chart with a fixed default $c$ selected the correct horizon on all
tested instances, reporting long horizons as uncertified in one stress case.
\end{abstract}

\begin{table}[t]
\noindent\textit{Nomenclature}\par\vspace{4pt}
{\small
\noindent\begin{tabular}{@{}l@{\hspace{6pt}}>{\raggedright\arraybackslash}p{0.66\columnwidth}@{}}
\hline
\multicolumn{2}{@{}l}{\emph{Notation}}\\
$(X)\succeq0, (X)\succ0$ & $(\bullet)$ being symmetric positive semidefinite, positive definite matrix\\
\multicolumn{2}{@{}l}{\emph{Horizon-Optimal LQR Parameters}}\\
$x_k,\ u_k$ \qquad \qquad & state and input \\ 
$A_k,\ B_k$ & time-varying dynamics\\
$Q_k \succeq0,\ Q_T \succeq0$ & stage and terminal state costs\\
$R_k \succ 0$ & input cost\\ 
$D_k$ & $B_kR_k^{-1}B_k\T$\\
$T,\ N,\ T^\star$ & horizon, maximum horizon, optimal horizon\\
$w,\ J_T$ & horizon penalty, optimal cost for horizon $T$\\
$\mathcal{N}(Q_k),\ \Pi_k$ & unpenalized subspace, its orthogonal projector\\
\multicolumn{2}{@{}l}{\emph{Coordinates}}\\
$P_k,\ \lambda_k$ & Riccati matrix, costate ($\lambda_k=P_kx_k$)\\
$\mathcal{L}_k$ & state--costate subspace $\operatorname{col}[I,\,P_k]\T$\\
$\Pt_k$ & inverse coordinate $P_k^{-1}$~\cite{hop}\\
$\Ph_k$ & shifted inverse coordinate $(P_k+cI)^{-1}$\\
$c,\ c_0$ & shift, default shift\\
$\varepsilon$ & diagonal loading, $Q_k\mapsto Q_k+\varepsilon I$\\
\multicolumn{2}{@{}l}{\emph{Composition}}\\
$E_k,F_k,G_k$ & inverse-coordinate stage coefficients\\
$\Eh_k,\Fh_k,\Gh_k$ & shifted stage coefficients\\
$E'_k,F'_k,G'_k$ & absorbed coefficients (Lemma~\ref{lem:absorb})\\
$\Eb_k,\Fb_k,\Gb_k$ & composed coefficients of stages $0,\dots,k$\\
$M_k$ & $I-c\Gh_k$, pivot (a)\\
$\phi,\ h_k,\ \psi_k$ & unshift map, shifted stage map, $\psi_k=h_k\circ\phi$\\
\multicolumn{2}{@{}l}{\emph{Certification}}\\
$\tau,\ \mathcal{C}$ & rcond threshold, certified horizons\\
\hline
\end{tabular}}
\end{table}
\section{INTRODUCTION}

Optimal control problems (OCPs) seek a trajectory of states and controls that
minimizes a cost functional subject to the dynamics. Common approaches such as
LQR~\cite{kirk2004optimal}, MPC~\cite{mpc}, iLQR~\cite{ilqr} and
DDP~\cite{ddp} use a fixed planning horizon. In many applications, however, a
shorter horizon reduces computation while a longer horizon lowers the task
cost, which has motivated treating the horizon length as a decision variable
of the OCP~\cite{ELALAMI1998525, Abdelhak01052010, teb}, including recent
schemes that adapt the horizon online from learned horizon
maps~\cite{gonzalez2023adaptive}.

For linear quadratic problems, finite-horizon LQR solves a fixed horizon by a
backward Riccati recursion~\cite{bellman1966dynamic}. A brute-force search
over $T=1,\dots,N$ repeats overlapping portions of this recursion, giving an
$N^2$ dependence on the maximum horizon $N$. For stationary $(A,B,Q,R)$,
\cite{op} exploits the fact that the per-step map does not depend on the
horizon: a single backward pass computes all Riccati matrices, and the cost
for any $T$ is read off by re-indexing. This property fails for time-varying
data, which is the generic case when nonlinear dynamics are linearized along
a trajectory.

A geometric view treats the Riccati matrix as a coordinate on a space of
subspaces. For a fixed horizon, the costate $\lambda_k$ is the gradient of
the cost-to-go $V_k(x)=\tfrac12x\T P_kx$, so along optimal trajectories
$\lambda_k=P_kx_k$. The state--costate pairs satisfying this relation form
the $n$-dimensional subspace
$\mathcal{L}_k=\operatorname{col}[I,\;P_k]^\top\subset\mathbb{R}^{2n}$, the
column space of the $2n\times n$ matrix $[I,\;P_k]^\top$, also called the
graph of $P_k$ (for the algebraic equation, the graph of
an invariant subspace of a symplectic pencil~\cite{sima2010structure}). It is a point of the Grassmannian $\mathrm{Gr}(n,2n)$, the
set of all $n$-dimensional subspaces of $\mathbb{R}^{2n}$. The Riccati
recursion propagates $\mathcal{L}_k$ backward in time through
$\mathrm{Gr}(n,2n)$, and $P_k$ is a local coordinate for it, a \emph{graph
coordinate}; see~\cite{sympletic} for the continuous-time analogue. The
same subspace admits other graph coordinates: when $P_k$ is invertible,
$\mathcal{L}_k=\operatorname{col}[P_k^{-1},\;I]^\top$, the graph of
$\widetilde{P}_k=P_k^{-1}$ over the costate. Under a linear change of the
state--costate variables, a graph coordinate undergoes a matrix linear
fractional transformation (LFT). The composition of~\cite{hop} uses the
inverse coordinate $\widetilde{P}_k$: in it, consecutive time-varying
Riccati maps compose into a single LFT, so all candidate horizons can be
queried after one forward composition pass, at $O(Nn^3)$ cost.

The inverse coordinate requires $Q_k^{-1}$ and $Q_T^{-1}$, whereas LQR
requires only $Q_k,Q_T\succeq0$. Singular costs arise whenever some states
are not penalized. Diagonal loading, $Q_k\mapsto Q_k+\varepsilon I$, restores
invertibility but changes the problem being solved, and with it the optimal
cost, the control law and possibly the selected horizon. Alternatively, lifting the problem over $d$ steps yields a positive definite
cost under uniform $d$-step observability and nonsingular
$A_k$~\cite{sun2023receding}; the shifted coordinate requires neither
assumption and leaves the stage structure unchanged.

The contributions of this work are:
\begin{enumerate}
\item \emph{Shifted inverse coordinate.} We represent the Riccati matrix by
$\widehat{P}_k=(P_k+cI)^{-1}$, the inverse coordinate after a fixed
symplectic shear, whose domain contains every positive semidefinite Riccati
matrix (Proposition~\ref{prop:domain}).
\item \emph{Absorption into the composition.} We show that the coordinate
change can be absorbed into the stage coefficients of the inverse-coordinate
LFT, so that horizon search keeps its $O(Nn^3)$ cost and, in exact
arithmetic, recovers the Riccati solution exactly, provided the invertibility
conditions of Assumption~\ref{as:admissibility} hold. These conditions hold
for all small $c$ when the dynamics contract the unpenalized directions, fail
for at most finitely many $c$ outside a degenerate case, and are certified
per horizon at run time.
\item \emph{Numerical comparison.} On random time-varying instances and
three physical systems with singular $Q_k$, we compare against diagonal
loading, an adaptive-jitter inverse coordinate and the released HOP
implementation, and report where each fails, including a stress case in
which the shifted chart reports some horizons as uncertified.
\end{enumerate}

Section~\ref{sec:formulation} states the problem and reviews the
inverse-coordinate composition, Section~\ref{sec:shifted} develops the
shifted chart and its admissibility, Section~\ref{sec:numerical} reports
the numerical study, and Section~V provides concluding remarks.

\section{Preliminaries and Problem Formulation}
\label{sec:formulation}
\subsection{Horizon-Optimal Time-Varying LQR}

Consider the discrete-time linear time-varying system
\begin{equation}
    x_{k+1} = A_k x_k + B_k u_k,
    \qquad k=0,\ldots,T-1,
\end{equation}
where $x_k \in \mathbb{R}^n$ and $u_k \in \mathbb{R}^m$, and $A_k$, $B_k$ are of compatible dimensions. The horizon $T$ is selected from
\begin{equation}
    T \in \{1,\ldots,N\},
\end{equation}
where $N$ denotes the maximum allowable horizon.

For a candidate horizon $T$, consider the quadratic objective
\begin{equation}
    J_T =
    \frac{1}{2}x_T^\top Q_T x_T
    +
    \frac{1}{2}
    \sum_{k=0}^{T-1}
    \left(
        x_k^\top Q_k x_k
        +
        u_k^\top R_k u_k
    \right)
    + wT,
    \label{eq:tv_lqr_cost}
\end{equation}
where $w \geq 0$ penalizes the horizon length. We assume
\begin{equation}
    R_k \succ 0,
\end{equation}
while $Q_k, Q_T \succeq 0$ need not be invertible; $Q_k,Q_T$ are $n\times n$ and $R_k$ is $m\times m$.

The horizon-optimal LQR problem is then
\begin{equation}
    T^\star
    =
    \underset{1\leq T\leq N}{\arg\min}\; J_T,
    \label{eq:holqr}
\end{equation}
where $J_T$ denotes the optimal value of~\eqref{eq:tv_lqr_cost} for a given candidate horizon $T$ and fixed initial state $x_0$.

\subsection{Finite-Horizon Riccati Recursion}
\label{sec:riccati}

For a fixed candidate horizon $T$, the finite-horizon LQR problem can be solved by dynamic programming. The optimal value function is quadratic,
\begin{equation}
    V_k(x_k)
    =
    \frac{1}{2}x_k^\top P_k x_k,
\end{equation}
with terminal condition
\begin{equation}
    P_T = Q_T.
    \label{eq:terminal_condition}
\end{equation}

The backward Riccati recursion is given by
\begin{align}
    S_k
    &=
    R_k + B_k^\top P_{k+1} B_k,
    \label{eq:riccati_S}
    \\
    K_k
    &=
    S_k^{-1} B_k^\top P_{k+1} A_k,
    \label{eq:riccati_K}
    \\
    P_k
    &=
    Q_k
    + A_k^\top P_{k+1} A_k
    - A_k^\top P_{k+1} B_k K_k.
    \label{eq:riccati}
\end{align}

The corresponding optimal control law is
\begin{equation}
    u_k^\star = -K_k x_k.
\end{equation}

The optimized cost for horizon $T$ is consequently
\begin{equation}
    J_T
    =
    \frac{1}{2}x_0^\top P_0 x_0 + wT.
    \label{eq:optimal_cost}
\end{equation}

The Riccati recursion~\eqref{eq:riccati} does not require $Q_k^{-1}$: since $R_k \succ 0$, the matrix $S_k$ in~\eqref{eq:riccati_S} is positive definite, and if $Q_k, Q_T \succeq 0$, then $P_k \succeq 0$ for all $k$. The fixed-horizon recursion is therefore well defined even when $Q_k$ is singular.

\subsection{Horizon Search and Computational Complexity}

A direct approach to~\eqref{eq:holqr} solves the finite-horizon LQR problem independently for every candidate horizon
$T=1,\ldots,N$. The candidate horizon $T$ requires $T$ backward Riccati steps, resulting in
\begin{equation}
    \sum_{T=1}^{N} T
    =
    \mathcal{O}(N^2)
\end{equation}
Riccati updates. Since each update requires $\mathcal{O}(n^3)$ matrix operations, the overall complexity of the brute-force horizon search is
\begin{equation}
    \mathcal{O}(N^2 n^3).
    \label{eq:bruteforce_complexity}
\end{equation}

\subsection{LFT Composition in the Inverse Coordinate}
\label{sec:lft}

Following HOP~\cite{hop}, define the inverse cost-to-go
\begin{equation}
    \widetilde{P}_k = P_k^{-1}.
\end{equation}

Applying the Woodbury Matrix Identity\cite{woodbury1950inverting} to the Riccati recursion\ref{eq:riccati} gives
\begin{equation}
    \widetilde{P}_k
    =
    E_k
    -
    F_k
    \left(
        \widetilde{P}_{k+1} + G_k
    \right)^{-1}
    F_k^\top,
    \label{eq:inverse_riccati}
\end{equation}
where
\begin{align}
    E_k
    &=
    Q_k^{-1},
    \label{eq:E}
    \\
    F_k
    &=
    Q_k^{-1} A_k^\top,
    \label{eq:F}
    \\
    G_k
    &=
    A_k Q_k^{-1} A_k^\top
    +
    B_k R_k^{-1} B_k^\top.
    \label{eq:G}
\end{align}

Thus, each Riccati step can be represented by the map
\begin{equation}
    \widetilde{g}_k(X)
    =
    E_k
    -
    F_k(X+G_k)^{-1}F_k^\top.
    \label{eq:lft_map}
\end{equation}

The maps in~\eqref{eq:lft_map} are closed under composition. In particular, the composition of successive Riccati maps can be represented using another triple of matrices,
\begin{equation}
    \widetilde{g}_{0:k}(X)
    =
    \bar{E}_k
    -
    \bar{F}_k
    (X+\bar{G}_k)^{-1}
    \bar{F}_k^\top.
    \label{eq:composed_lft}
\end{equation}

Starting from
\begin{equation}
    (\bar{E}_0,\bar{F}_0,\bar{G}_0)
    =
    (E_0,F_0,G_0),
\end{equation}
the composed parameters can be updated recursively as
\begin{align}
    W_k
    &=
    (E_k+\bar{G}_{k-1})^{-1},
    \label{eq:W}
    \\
    \bar{E}_k
    &=
    \bar{E}_{k-1}
    -
    \bar{F}_{k-1}W_k\bar{F}_{k-1}^\top,
    \label{eq:Ebar}
    \\
    \bar{F}_k
    &=
    \bar{F}_{k-1}W_kF_k,
    \label{eq:Fbar}
    \\
    \bar{G}_k
    &=
    G_k-F_k^\top W_kF_k.
    \label{eq:Gbar}
\end{align}

For a candidate horizon $T$, the terminal inverse cost is
\begin{equation}
    \widetilde{P}_T = Q_T^{-1},
\end{equation}
and the initial inverse cost-to-go can be obtained from the composed map as
\begin{equation}
    \widetilde{P}_0^{(T)}
    =
    \widetilde{g}_{0:T-1}
    \left(
        Q_T^{-1}
    \right).
    \label{eq:horizon_query}
\end{equation}

The corresponding horizon cost is recovered as
\begin{equation}
    J_T
    =
    \frac{1}{2}
    x_0^\top
    \left(
        \widetilde{P}_0^{(T)}
    \right)^{-1}
    x_0
    +wT.
    \label{eq:hop_cost}
\end{equation}

The composed maps can be constructed in a single forward pass over the horizon. Consequently, the horizon-dependent Riccati information can be reused rather than recomputed independently for every candidate horizon, giving an overall complexity of
\begin{equation}
    \mathcal{O}(Nn^3),
    \label{eq:hop_complexity}
\end{equation}
compared with the $\mathcal{O}(N^2n^3)$ complexity of the brute-force approach.

The coefficients $(E_k, F_k, G_k)$ and the terminal query require $Q_k^{-1}$ and $Q_T^{-1}$, which excludes singular state costs. The next section introduces a shifted inverse coordinate that removes this requirement while recovering the Riccati matrix exactly.

\section{Shifted-Chart Composition}
\label{sec:shifted}
The inverse coordinate of Section~\ref{sec:lft} requires $Q_k^{-1}$ and $Q_T^{-1}$, although LQR permits $Q_k \succeq 0$. Rather than altering the cost by diagonal loading, we change the coordinate in which the
state-costate subspace is represented, keeping the three-parameter LFT structure that composition requires.

\subsection{Shifted Inverse Coordinate}
\label{sec:shifted_coordinate}
The state-costate subspace associated with a Riccati matrix can be written
as the graph
\begin{equation}
    \mathcal{L}_k
    =
    \operatorname{col}
    \begin{bmatrix}
        I\\
        P_k
    \end{bmatrix}.
    \label{eq:riccati_subspace_shifted}
\end{equation}
When $P_k$ is nonsingular, the same subspace can equivalently be
represented as
\begin{equation}
    \mathcal{L}_k
    =
    \operatorname{col}
    \begin{bmatrix}
        P_k^{-1}\\
        I
    \end{bmatrix},
    \label{eq:inverse_graph}
\end{equation}
which corresponds to describing the subspace as a graph over the costate, $x_k=P_k^{-1}\lambda_k$. Thus, the inverse Riccati coordinate used in Section~\ref{sec:lft} is a coordinate representation of
the same state-costate subspace, but it is defined only when $P_k$ is invertible.

We instead apply a fixed shear to the state-costate variables,
\begin{equation}
    S_c:
    (x,\lambda)
    \mapsto
    (x,\lambda+cx),
    \qquad c>0,
\end{equation}
with matrix representation
\begin{equation}
    S_c
    =
    \begin{bmatrix}
        I & 0\\
        cI & I
    \end{bmatrix}.
    \label{eq:shear_matrix}
\end{equation}
Since $cI$ is symmetric, $S_c$ is symplectic. Its action on the Riccati
graph is
\begin{equation}
    S_c
    \begin{bmatrix}
        I\\
        P_k
    \end{bmatrix}
    =
    \begin{bmatrix}
        I\\
        P_k+cI
    \end{bmatrix}.
    \label{eq:sheared_graph}
\end{equation}
Taking the inverse graph coordinate in these sheared coordinates gives the shifted inverse coordinate
\begin{equation}
    \hat{P}_k
    :=
    (P_k+cI)^{-1},
    \qquad
    P_k
    =
    \hat{P}_k^{-1}-cI.
    \label{eq:shifted_coord_def}
\end{equation}
Thus, the proposed coordinate is the inverse coordinate applied after a fixed symplectic shear of the state-costate variables. The underlying Riccati matrix, and hence the original LQR cost-to-go is recovered exactly.

\begin{proposition}[Chart admissibility]
\label{prop:domain}
Every $P\succeq0$ lies in the domain of $\widehat{P}=(P+C)^{-1}$ if and only
if $C\succ0$. In particular, for $C=cI$ with $c>0$, all Riccati iterates
$P_k$ and all $Q_k+cI$, $Q_T+cI$ are invertible, with
$\lambda_{\min}(P_k+cI)\ge c$.
\end{proposition}
\begin{proof}
If $C\succ0$, then $v\T(P+C)v\ge v\T Cv>0$ for $v\ne0$. Conversely, if
$Cv=\lambda v$ with $v\ne0$ and $\lambda\le0$, then $P=-\lambda vv\T\succeq0$
and $(P+C)v=0$.
\end{proof}
 
\subsection{Shifted Stage Map and LFT Absorption}
\label{sec:absorption}

For invertible $P_{k+1}$, write $\widetilde{P}_{k+1} = P_{k+1}^{-1}$. The
Woodbury identity rewrites the Riccati recursion~\eqref{eq:riccati} as
$P_k = Q_k + A_k\T\big(\widetilde{P}_{k+1} + B_kR_k^{-1}B_k\T\big)^{-1}A_k$.
Adding $cI$ to both sides and inverting once more with Woodbury gives
\begin{equation}
    \widehat{P}_k = h_k(\widetilde{P}_{k+1})
    := \widehat{E}_k - \widehat{F}_k\big(\widetilde{P}_{k+1} + \widehat{G}_k\big)^{-1}\widehat{F}_k\T,
    \label{eq:h_map}
\end{equation}
with stage coefficients that require only $Q_k + cI \succ 0$:
\begin{align}
    \widehat{E}_k &= (Q_k + cI)^{-1}, \label{eq:Ehat_def}\\
    \widehat{F}_k &= \widehat{E}_k A_k\T, \label{eq:Fhat_def}\\
    \widehat{G}_k &= A_k\widehat{E}_kA_k\T + B_kR_k^{-1}B_k\T. \label{eq:Ghat_def}
\end{align}
The map $h_k$ consumes the unshifted inverse $\widetilde{P}_{k+1}$, whereas
the recursion carries $\widehat{P}_{k+1}$. By~\eqref{eq:shifted_coord_def}
the two are related by
\begin{equation}
    \widetilde{P} = \big(\widehat{P}^{-1} - cI\big)^{-1}
    = \widehat{P}\big(I - c\widehat{P}\big)^{-1} =: \phi(\widehat{P}),
    \label{eq:phi_map_def}
\end{equation}
so the one-step update in shifted coordinates is $\psi_k := h_k\circ\phi$.
The following lemma absorbs $\phi$ into the stage coefficients.

\begin{lemma}[Absorption]
\label{lem:absorb}
Let $\mathcal{R}_k$ denote the Riccati map~\eqref{eq:riccati},
$P_k=\mathcal{R}_k(P_{k+1})$, and define the exact shifted step
\begin{equation}
  \psi_k(X) := \big(\mathcal{R}_k(X^{-1}-cI)+cI\big)^{-1},~ X=(P+cI)^{-1},\ P\succeq0 .
  \label{eq:psi_def}
\end{equation}
Assume $M_k:=I-c\widehat{G}_k$ is invertible and let
\begin{align}
  E'_k &= \widehat{E}_k + c\,\widehat{F}_k M_k^{-1}\widehat{F}_k\T, \label{eq:Eprime_def}\\
  F'_k &= \widehat{F}_k M_k^{-1}, \label{eq:Fprime_def}\\
  G'_k &= \widehat{G}_k M_k^{-1}. \label{eq:Gprime_def}
\end{align}
Then $E'_k$ and $G'_k$ are symmetric, and
\begin{equation}
  \psi_k(X) = E'_k - F'_k\,(X+G'_k)^{-1}{F'_k}\T
  \label{eq:psi}
\end{equation}
for every $X$ in~\eqref{eq:psi_def} with $X+G'_k$ invertible.
\end{lemma}

\begin{proof}
Drop the index $k$ and let $M=I-c\Gh$. As $M$ is a polynomial in the
symmetric $\Gh$, $M^{-1}$ is symmetric and commutes with $\Gh$, so $E'$ and
$G'$ are symmetric. First let $P\succ0$; by the derivation above,
$\psi=h\circ\phi$. Since $X$ commutes with $(I-cX)^{-1}$ and
$\Gh=\Gh M^{-1}M$,
\begin{align*}
\phi(X)+\Gh&=(I-cX)^{-1}\big[X+(I-cX)\Gh\big]\\
&=(I-cX)^{-1}(X+G')M.
\end{align*}
Since $M^{-1}-cG'=(I-c\Gh)M^{-1}=I$, we have $I-cX=M^{-1}-c(X+G')$, so
\[
\big(\phi(X)+\Gh\big)^{-1}=M^{-1}(X+G')^{-1}M^{-1}-cM^{-1}.
\]
Substituting into $h$ and using $M^{-1}\Fh\T=F'^\top$,
\begin{align*}
\psi(X)&=\Eh+c\Fh M^{-1}\Fh\T-F'(X+G')^{-1}F'^\top\\
&=E'-F'(X+G')^{-1}F'^\top.
\end{align*}

For singular $P\succeq0$, let $P_j=P+I/j\succ0$ and
$X_j=(P_j+cI)^{-1}\to X$. The left side
$\psi(X_j)=(\mathcal{R}(P_j)+cI)^{-1}$ converges to $\psi(X)$, since
$\mathcal{R}$ is continuous on $P\succeq0$ (as $R\succ0$) and
$\mathcal{R}(P)+cI\succeq cI$. The right side converges because $X+G'$ is
invertible. Letting $j\to\infty$ gives the identity at $P$.

\end{proof}

\begin{remark}[Matrix Shifts]
\label{rem:general_shift}
Lemma~\ref{lem:absorb} extends to symmetric shears $S_C = \left[\begin{matrix} I & 0 \\ C & I \end{matrix}\right]$ with symmetric $C$, giving $E' = \widehat{E} + \widehat{F} C(I - \widehat{G}C)^{-1}\widehat{F}^\top$, $F' = \widehat{F}(I - C\widehat{G})^{-1}$, and $G' = (I - \widehat{G}C)^{-1}\widehat{G}$. We chose $C = cI$ to avoid non-commutative factor ordering and reduce parameter tuning to scalar $c$.
\end{remark}

\subsection{Admissibility and Choice of the Shift}
\label{sec:admissibility}
Proposition~\ref{prop:domain} guarantees that the shifted \emph{chart} contains every Riccati matrix for any $c>0$. It does not guarantee that the stage \emph{maps} keep the three-parameter form with finite coefficients, nor that their composition and the horizon queries do. Each condition below is the pivot of a composition in family~\eqref{eq:lft_map}: of $h_k$ with $\phi$ for (a), of the prefix with stage $k$ for (b), and of the prefix with the terminal value for (c).

\begin{assumption}[Admissibility]
\label{as:admissibility}
For a given $c$: (a) $I-c\widehat{G}_k$ is invertible for all $k$;
(b) $E'_k+\bar{G}_{k-1}$ is invertible for all $k$;
(c) $(Q_T^{(t)}+cI)^{-1}+\bar{G}_{t-1}$ is invertible for all $t$.
\end{assumption}

For $Q_k,Q_T\succeq0$, the shifted coordinate and $\Eh_k$ exist for every
$c>0$ (Proposition~\ref{prop:domain}); only the stage maps and their
composition can fail, and they can fail for every $c$
(Remark~\ref{rem:persistent}).

\begin{theorem}[Shifted-chart horizon search]
\label{thm:composition}
Let $Q_k,Q_T\succeq0$, $R_k\succ0$ and $c>0$, and let
Assumption~\ref{as:admissibility} hold. Then the triples $(E'_k,F'_k,G'_k)$
compose by~\eqref{eq:W}--\eqref{eq:Gbar} with
$(E_k,F_k,G_k)\to(E'_k,F'_k,G'_k)$, and Algorithm~\ref{alg:shop} returns, in
exact arithmetic, the optimal cost $J_T$ for every $T\in\{1,\dots,N\}$ in
$O(Nn^3)$ operations, without inverting $Q_k$ or $Q_T$ and without
modifying the cost.
\end{theorem}
\begin{proof}
By Proposition~\ref{prop:domain} the coefficients
$\widehat{E}_k,\widehat{F}_k,\widehat{G}_k$ exist, and by
Lemma~\ref{lem:absorb} each exact shifted step $\psi_k$, including at
singular $P_{k+1}$, lies in the family~\eqref{eq:lft_map}. The composition
rule~\eqref{eq:W}--\eqref{eq:Gbar} uses only this form, so it applies
unchanged and gives $\widehat{P}_0^{(T)}=\psi_{0:T-1}\big((Q_T+cI)^{-1}\big)$;
then $P_0=(\widehat{P}_0^{(T)})^{-1}-cI$ by~\eqref{eq:shifted_coord_def}.
The $N$ stage constructions, $N-1$ composition steps and $N$ queries each
cost $O(n^3)$.
\end{proof}

\begin{proposition}[Generic admissibility]
\label{prop:generic}
For fixed data, the determinants of the matrices in
Assumption~\ref{as:admissibility} are rational functions of $c$ on $c>0$.
Unless one vanishes identically, Assumption~\ref{as:admissibility} fails for
at most finitely many $c>0$.
\end{proposition}
\begin{proof}
$\widehat{E}_k=(Q_k+cI)^{-1}$ is rational in $c$ with poles only at
$-\lambda_i(Q_k)\le0$, so $\widehat{F}_k,\widehat{G}_k$, and wherever (a)
holds $E'_k,F'_k,G'_k$, are rational on $c>0$. By induction
through~\eqref{eq:W}--\eqref{eq:Gbar}, so are $\bar{E}_k,\bar{F}_k,\bar{G}_k$
and the matrix in (c). A nonzero rational function has finitely many zeros.
\end{proof}

\begin{remark}[Persistent failure]\label{rem:persistent}
If $w\ne0$ satisfies $B_k\T w=0$, $Q_kA_k\T w=0$ and $A_kA_k\T w=w$, then
$\Eh_kA_k\T w=A_k\T w/c$ and $(I-c\Gh_k)w=0$ for every $c>0$. This is the
degenerate case of Proposition~\ref{prop:generic}; it has
$\lVert A_k\Pi_k\rVert_2\ge1$.
\end{remark}
\begin{remark}[Factorization]\label{rem:ldl}
The pivots are symmetric but, outside the regime of
Proposition~\ref{prop:sufficient}, indefinite. They are factored by LU or
$LDL\T$, not Cholesky.
\end{remark}

\begin{proposition}[Sufficient condition]
\label{prop:sufficient}
Let $D_k=B_kR_k^{-1}B_k\T$ and let $\Pi_k$ project onto $\mathcal{N}(Q_k)$.
(i) If $c\,\lambda_{\max}(\Gh_k)<1$ for all $k$, every pivot is symmetric
positive definite. (ii) Such $c>0$ exists if and only if
$\lVert A_k\Pi_k\rVert_2<1$ for all $k$; then (i) holds for all
$c\in(0,\bar c)$, with
$\bar c=\min_k\big(1-\lVert A_k\Pi_k\rVert_2^2\big)\big/
\big(\lVert A_k(I-\Pi_k)\rVert_2^2/\lambda_+(Q_k)+\lVert D_k\rVert_2\big)$
and $\lambda_+$ the smallest positive eigenvalue.
\end{proposition}
\begin{proof}
\emph{(i)} Drop the index $k$ and let $M=I-c\Gh\succ0$ and $S=A\Eh A\T$, so
that $\Gh=S+D$ and $M+cS=I-cD\succ0$. Then $E'\succ0$ and
$G'=\Gh M^{-1}\succeq0$. Multiplying $E'$ by $\Fh\T\Eh^{-1}$ gives
$\Fh\T\Eh^{-1}E'=(I+cSM^{-1})\Fh\T$, from which
\begin{align*}
F'^\top E'^{-1}F' &= (I-cD)^{-1}SM^{-1},\\
G'-F'^\top E'^{-1}F' &= (I-cD)^{-1}D \succeq 0 .
\end{align*}
If $\Gb_{k-1}\succeq0$, then $E'_k+\Gb_{k-1}\succeq E'_k\succ0$, so
$W_k\preceq E_k'^{-1}$ and
\[
\Gb_k \succeq G'_k-F_k'^\top E_k'^{-1}F'_k \succeq 0 .
\]
Since $\Gb_0=G'_0\succeq0$, this holds for all $k$, and pivot (c) is
positive definite.

\emph{(ii)} The matrix $c\Eh_k=(I+Q_k/c)^{-1}$ has eigenvalue $1$ on
$\mathcal{N}(Q_k)$ and $c/(c+\lambda_i)\le c/\lambda_+$ otherwise, so
\[
\Pi_k \preceq c\Eh_k \preceq \Pi_k+\tfrac{c}{\lambda_+}(I-\Pi_k).
\]
Substituting into $c\Gh_k=A_k(c\Eh_k)A_k\T+cD_k$ gives
\begin{align*}
c\lambda_{\max}(\Gh_k) &\ge \lVert A_k\Pi_k\rVert^2,\\
c\lambda_{\max}(\Gh_k) &\le \lVert A_k\Pi_k\rVert^2
 + c\big(\lVert A_k(I-\Pi_k)\rVert^2/\lambda_+ + \lVert D_k\rVert\big),
\end{align*}
and $c\Gh_k$ is nondecreasing in $c$.
\end{proof}
Under (i), the shifted maps are inverse-coordinate maps of an auxiliary
problem with positive definite cost $E_k'^{-1}$, i.e., the setting
of~\cite{hop}. The physical systems of Section~\ref{sec:numerical} have
$\lVert A_k\Pi_k\rVert_2\ge1$; there the pivots are indefinite
(Remark~\ref{rem:ldl}), and admissibility rests on
Proposition~\ref{prop:generic} and run-time certification.

\emph{Run-time certification.} Algorithm~\ref{alg:shop} estimates the
reciprocal condition number of every pivot from its LU factors, at $O(n^2)$
cost each. Horizon $t$ is \emph{certified} if every pivot it uses, namely
(a) and (b) at stages $0,\dots,t-1$ and (c) at $t$, has reciprocal
condition number at least $\tau$. Uncertified horizons may be retried with
$c\gets\rho c$; a retry can only add certified horizons. The algorithm
returns $J_t$ and the minimizer over certified horizons only, with the
certified set, so a violation of Assumption~\ref{as:admissibility} is
reported rather than returned as a cost, and the $O(Nn^3)$ bound is
unchanged. Certification bounds the conditioning of each pivot, not the
final error: near the threshold the error can approach $u/\tau$, so $\tau$
trades coverage against accuracy.

\begin{algorithm}[t]
\caption{Shifted HOP-LQR with certification}
\label{alg:shop}
\begin{algorithmic}[1]
\footnotesize
\Require $\{A_k,B_k,Q_k,R_k\}_{k=0}^{N-1}$, $\{Q_T^{(t)}\}_{t=1}^{N}$, $x_0$, $w$, $c$, $\tau$
\Ensure certified set $\mathcal{C}$, $\{J_t\}_{t\in\mathcal{C}}$, $T^\star=\arg\min_{t\in\mathcal{C}}J_t$
\Statex \textit{Stage coefficients (Lemma~\ref{lem:absorb})}
\For{$k=0$ \textbf{to} $N-1$}
  \State $\Eh_k\gets(Q_k+cI)^{-1}$;\; $\Fh_k\gets\Eh_kA_k\T$
  \State $\Gh_k\gets A_k\Eh_kA_k\T+B_kR_k^{-1}B_k\T$
  \State $M_k\gets I-c\Gh_k$;\; $r^a_k\gets\mathrm{rcond}(M_k)$ \Comment{(a)}
  \State $F'_k\gets\Fh_kM_k^{-1}$;\; $G'_k\gets\Gh_kM_k^{-1}$
  \State $E'_k\gets\Eh_k+cF'_k\Fh_k\T$
\EndFor
\Statex \textit{Composition}
\State $(\Eb_0,\Fb_0,\Gb_0)\gets(E'_0,F'_0,G'_0)$;\; $s_1\gets r^a_0$
\For{$k=1$ \textbf{to} $N-1$}
  \State $r^b_k\gets\mathrm{rcond}(E'_k+\Gb_{k-1})$ \Comment{(b)}
  \State update $(\Eb_k,\Fb_k,\Gb_k)$ by~\eqref{eq:W}--\eqref{eq:Gbar}
  \State $s_{k+1}\gets\min(s_k,\,r^a_k,\,r^b_k)$
\EndFor
\Statex \textit{Queries and certification}
\State $\mathcal{C}\gets\emptyset$
\For{$t=1$ \textbf{to} $N$}
  \State $H_t\gets(Q_T^{(t)}+cI)^{-1}+\Gb_{t-1}$;\; $r^c_t\gets\mathrm{rcond}(H_t)$ \Comment{(c)}
  \State $\Ph_0\gets\Eb_{t-1}-\Fb_{t-1}H_t^{-1}\Fb_{t-1}\T$
  \State $J_t\gets\tfrac12x_0\T(\Ph_0^{-1}-cI)x_0+wt$
  \State \textbf{if} $\min(s_t,r^c_t)\ge\tau$ \textbf{then} $\mathcal{C}\gets\mathcal{C}\cup\{t\}$
\EndFor
\end{algorithmic}
\end{algorithm}

%

\section{Numerical Study}
\label{sec:numerical}

Errors are $\max_{1\le t\le N}|J_t-J_t^{\rm ref}|/\max(|J_t^{\rm ref}|,10^{-12})$,
the worst disagreement over all candidate horizons, and a horizon is
\emph{wrong} when the method's $T^\star$ differs from brute force. Unless
stated otherwise, Shifted-HOP runs Algorithm~\ref{alg:shop} with $c=c_0=1$,
the checks of Sec.~\ref{sec:admissibility} enabled ($\tau=10^{-10}$), and no
tuning. Loading denotes the inverse coordinate with
$Q_k\mapsto Q_k+\varepsilon I$; \emph{oracle} $\varepsilon$ is the grid value
($10^{-20}$--$10^{2}$) minimizing cost error against brute force, which is
unavailable in practice. Naive HOP is the inverse coordinate with an
escalating Cholesky jitter starting at $10^{-10}$.

\subsection{Perturbing \texorpdfstring{$Q_k$}{Qk}}
\label{sec:loading}

\begin{figure}[t]
\centering
\includegraphics[width=\columnwidth]{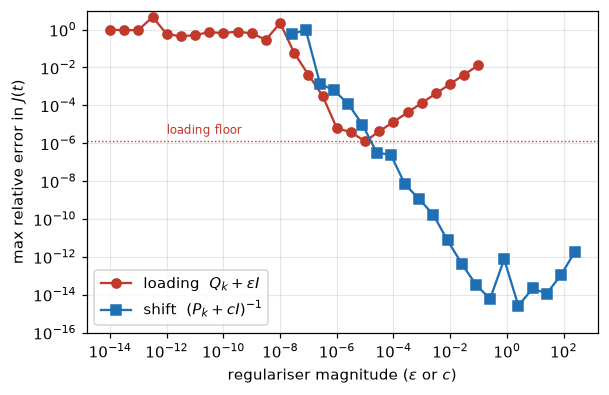}
\caption{Left: relative error against regularizer magnitude ($n=6$,
$N=18$); loading is $U$-shaped with a floor near $10^{-6}$, the shifted chart
is flat over a wide range of $c$. Right: best-of-grid error against the
coupling of an unpenalized direction ($n=6$, $N=30$, median of 4 seeds).}
\label{fig:loading}
\end{figure}

Diagonal loading restores invertibility by changing the problem.
Fig.~\ref{fig:loading} (left) compares both methods at their best parameter:
loading reaches $2.7\times10^{-6}$ at $\varepsilon=1.4\times10^{-6}$, while
the shifted chart reaches $4.3\times10^{-15}$ at its best $c$ and
$1.3\times10^{-14}$ at $c_0$. Too small an $\varepsilon$ amplifies
$(Q_k+\varepsilon I)^{-1}$ through the composition; too large an
$\varepsilon$ solves a different problem. Across four configurations
($n\in\{4,6,8\}$, $N\in\{20,60\}$, 8 seeds each) the shifted chart wins all
$32$ trials, at its best $c$ and at $c_0$ alike, with median error ratios
$2.7$--$5.1\times10^{8}$.

\begin{table}[t]
\caption{Horizon selection and cost error without tuning. Random: $200$
instances $\times$ $5$ draws of $x_0$ (wrong out of $1000$, max.\ relative
error). Free mode: $\rho(A_k)=1.05$ plus an unstable, uncontrollable,
unpenalized mode, $n\in\{6,12,24\}$ (wrong out of $15$).}
\label{tab:eps}
\centering
\small
\setlength{\tabcolsep}{3.5pt}
\begin{tabular}{lrrrr}
\hline
 & \multicolumn{2}{c}{Random} & \multicolumn{2}{c}{Free mode} \\
Method & wrong & max err & $N{=}120$ & $N{\ge}500$ \\
\hline
Loading, $\varepsilon=10^{-4}$  & 0   & $3.4\times10^{-2}$ & 0  & 0 \\
Loading, $\varepsilon=10^{-6}$  & 0   & $3.5\times10^{-4}$ & 0  & 15 \\
Loading, $\varepsilon=10^{-8}$  & 46  & $3.9\times10^{2}$  & 2  & 2 \\
Loading, $\varepsilon=10^{-10}$ & 221 & $2.5\times10^{3}$  & 15 & 15 \\
Loading, oracle $\varepsilon$   & 0   & $1.3\times10^{-4}$ & 0  & 11 \\
Naive HOP                       & 240 & $1.0\times10^{2}$  & 13 & 13 \\
Shifted, $c_0=1$                & 0   & $1.0\times10^{-11}$ & 0 & 0$^\dagger$ \\
\hline
\multicolumn{5}{l}{\footnotesize $^\dagger$Horizons beyond $t\approx215$--$231$ uncertified; $T^\star$ certified on all 15.}
\end{tabular}
\end{table}

The floor depends on how the dynamics couple unpenalized and penalized
states. Let $\mathcal{N}(Q_k)=\operatorname{span}\{v\}$, $\lVert v\rVert=1$,
and let $\kappa=\lVert A_kv-(v\T A_kv)v\rVert$ be the part of $A_kv$ that
leaves this direction. As $\kappa$ grows from $0$ to $2.17$, loading's best
error rises from $5.0\times10^{-15}$ to $4.6\times10^{-6}$, while Shifted-HOP
stays at or below $4.2\times10^{-15}$ (Fig.~\ref{fig:coupling}).

In practice $\varepsilon$ must be fixed without a reference.
Table~\ref{tab:eps} tests four fixed values on random instances and on
\emph{free-mode} instances, which add an unstable, uncontrollable,
unpenalized mode. No fixed $\varepsilon$ is both accurate and reliable:
$\varepsilon\le10^{-8}$ picks the wrong horizon on up to $22\%$ of random
trials, $\varepsilon=10^{-6}$ on all $15$ free-mode instances with
$N\ge500$, and only $\varepsilon=10^{-4}$ is correct on both sets, at cost
errors up to $1.4\times10^{10}$. Adaptive jitter is wrong on $24\%$ of
random trials, and on $24.7\%$ when restricted to $T\ge2$. Shifted-HOP with
$c_0=1$ is correct on every instance of both sets.
\begin{figure}[t]
\centering
\includegraphics[width=\columnwidth]{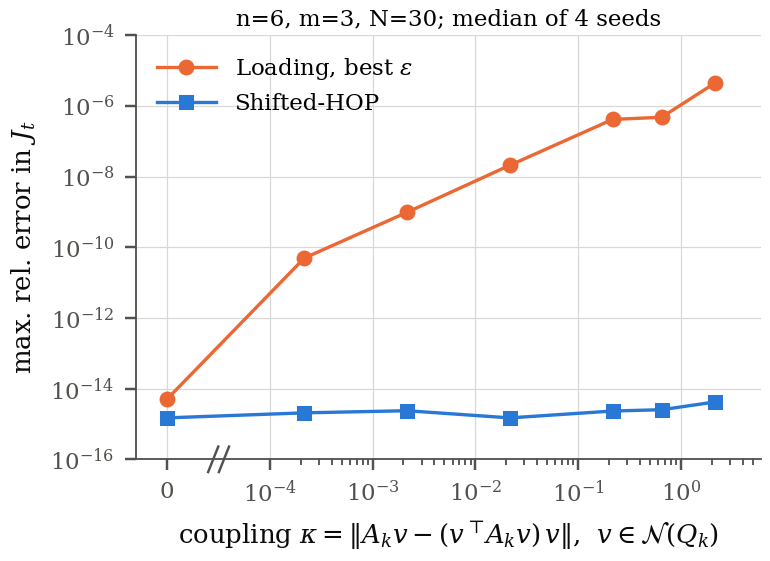}
\caption{Best-of-grid relative cost error against the coupling $\kappa$ of an
unpenalized direction ($n=6$, $m=3$, $N=30$, median of $4$ seeds).}
\label{fig:coupling}
\end{figure}

\subsection{Admissibility in Practice}
\label{sec:adm_numerics}

The checks matter only when a pivot degrades, and they turn a silent failure
into a reported one. On the free-mode instances of Table~\ref{tab:eps}, the
reciprocal condition number of pivot (c) decays with the horizon to
$10^{-18}$. Without the checks, Algorithm~\ref{alg:shop} returns every cost
anyway, with relative errors of $4\times10^{5}$ at $N=500$ and
$4\times10^{13}$ at $N=1000$, and selects the wrong horizon on $2$ and $7$
of $15$ instances. With the checks, horizons beyond $t\approx215$--$231$ are
reported as uncertified, the true $T^\star$ is certified on every instance,
and the error on certified horizons is at most $1.5\times10^{-5}$. On all
other runs ($346$), every horizon was certified without a retry; the
smallest reciprocal condition numbers were $2.3\times10^{-7}$,
$3.0\times10^{-8}$ and $1.9\times10^{-7}$ for pivots (a), (b) and (c).

\subsection{The Released HOP Implementation}
\label{sec:di}

The released HOP-DDP code~\cite{hop} does not jitter $Q_k$; it builds an
augmented $(n{+}1)$-dimensional system per stage and jitters that matrix. On a double integrator
($Q_k=\mathrm{diag}(1,0)$), a segway and a cart with time-varying mass, all
with singular $Q_k$, it selects the correct horizon on every trial when
$Q_T$ shares the running cost's null space. With a full-rank terminal cost,
and the search over all $T\in\{1,\dots,N\}$ as in problem~\eqref{eq:holqr},
it is wrong on $24.5$--$34.5\%$ of trials; every wrong pick is $T^\star\le2$,
and restricting the search to $T\ge2$ removes all failures. These counts are
sensitive to the last bit of $x_0$: perturbing $x_0$ by one to four ulp
changes \textsc{HOP (ref)}'s selection on $175$ of $200$ segway trials and
moves its failure count between $48$ and $69$, while brute force changes on
none. Shifted-HOP, Naive HOP and Loading are correct on every trial of this
protocol. The authors' own examples search $T\ge10$ and are unaffected.

\subsection{Exactness and Scaling}
\label{sec:runtime}

On the random instances, Shifted-HOP and an independent condensed-QP
formulation agree with brute force on every horizon choice; the median over
$x_0$ draws of the worst cost error is $2.4\times10^{-12}$ for Shifted-HOP
and $1.8\times10^{-14}$ for the QP. The true optimum is strictly interior in
$38\%$ of trials. The identities of Lemma~\ref{lem:absorb} and
Theorem~\ref{thm:composition} hold to within $1.5\times10^{-11}$ on $200$
singular instances. Error does not grow with the horizon: up to $N=1000$ it
stays below $10^{-12}$ on contracting random dynamics and below
$1.7\times10^{-10}$ on the open-loop unstable segway ($\rho(A)=1.34$) and on
random dynamics with $\rho(A_k)=1.05$.

\begin{table}[t]
\caption{Median runtime (ms), $n=6$, $m=2$, $15$ repeats;
Shifted-HOP includes $c_0$ and the admissibility checks.}
\label{tab:runtime}
\centering
\small
\setlength{\tabcolsep}{4pt}
\begin{tabular}{rrrrrr}
\hline
$N$ & BF & Shifted & Naive & Loading & HOP (ref) \\
\hline
10  & 0.99 & 1.10 & 1.12 & 0.73 & 2.17 \\
40  & 14.7 & 4.37 & 4.97 & 3.32 & 10.2 \\
120 & 121  & 12.2 & 13.0 & 8.64 & 27.0 \\
\hline
\end{tabular}
\end{table}

A fit $t=aN^p+b$ over $N\in\{10,\dots,120\}$ gives $p=1.98$ for brute force
and $p\le1.0$ for every composition method (Table~\ref{tab:runtime}).
Shifted-HOP runs at $0.88$--$0.98\times$ Naive HOP's time and
$1.3$--$1.5\times$ Loading's; the admissibility checks account for
$26$--$29\%$ of its runtime at $N=120$. \textsc{HOP (ref)} is
$2.0$--$2.3\times$ slower than Shifted-HOP. On the three physical systems at
$N=30$, Shifted-HOP is $2.6$--$2.8\times$ faster than brute force.

\section{Conclusion}
We introduced a shifted inverse coordinate for horizon-optimal time-varying
LQR that admits positive semidefinite state costs, and showed that it can be
absorbed into the inverse-coordinate LFT composition while retaining
$O(Nn^3)$ horizon search and, in exact arithmetic, exact recovery of the
Riccati solution. Admissibility of the shifted stage maps holds for all small
$c$ when the dynamics contract the unpenalized directions; otherwise it fails
for at most finitely many $c$ outside a degenerate case. A run-time check
certifies each horizon from the existing factorizations and reports
uncertified horizons instead of returning their costs. In our experiments,
every horizon was certified on the random and physical test problems. On
instances with an unstable, uncontrollable, unpenalized mode, long horizons
were uncertified while the optimal horizon remained certified. With the
checks enabled, the method ran at the speed of an adaptive-jitter inverse
coordinate and slower than diagonal loading.

The analysis is limited to linear time-varying LQR with positive
semidefinite costs. Certification bounds the conditioning of each pivot, not
the final error. Indefinite costs, as arise from nonconvex running costs in
DDP, additionally require $P_k\succ-cI$, which a fixed $c$ does not
guarantee. Embedding the construction in DDP for nonlinear systems, and
extending it to constrained problems~\cite{ipddp}, are left to future work.

\section*{Acknowledgement}
The authors acknowledge the use of LLMs for improving
the grammatical structure and clarity of this manuscript.


\bibliographystyle{IEEEtran}
\bibliography{references}

@inproceedings{hop,
author = {Dai, Miaomiao and Ren, Zhongqiang},
booktitle = {Robotics: Science and Systems (RSS)},
year = {2026},
title = {HOP: Fast Differential Dynamic Programming for Horizon-Optimal Trajectory Planning},
doi = {10.15607/RSS.2026.XXII.186}
}

@article{ddp,
author = {DAVID MAYNE},
title = {A Second-order Gradient Method for Determining Optimal Trajectories of Non-linear Discrete-time Systems},
journal = {International Journal of Control},
volume = {3},
number = {1},
pages = {85--95},
year = {1966},
publisher = {Taylor \& Francis},
doi = {10.1080/00207176608921369},
URL = { 
        https://doi.org/10.1080/00207176608921369
},
eprint = { 
        https://doi.org/10.1080/00207176608921369
}

}

@inproceedings{ilqr,
  title={Iterative linear quadratic regulator design for nonlinear biological movement systems},
  author={Li, Weiwei and Todorov, Emanuel},
  booktitle={First International Conference on Informatics in Control, Automation and Robotics},
  volume={2},
  pages={222--229},
  year={2004},
  organization={SciTePress}
}

@INPROCEEDINGS{op,
  author={Stachowicz, Kyle and Theodorou, Evangelos A.},
  booktitle={2022 International Conference on Robotics and Automation (ICRA)}, 
  title={Optimal-Horizon Model Predictive Control with Differential Dynamic Programming}, 
  year={2022},
  volume={},
  number={},
  pages={1440-1446},
  doi={10.1109/ICRA46639.2022.9812036}}

@book{mpc,
  author    = {Rawlings, James B. and Mayne, David Q. and Diehl, Moritz M.},
  title     = {Model Predictive Control: Theory, Computation, and Design},
  edition   = {2nd},
  publisher = {Nob Hill Publishing},
  address   = {Madison, WI},
  year      = {2017}
}

@INPROCEEDINGS{teb,
  author={Rösmann, Christoph and Hoffmann, Frank and Bertram, Torsten},
  booktitle={2015 European Control Conference (ECC)}, 
  title={Timed-Elastic-Bands for time-optimal point-to-point nonlinear model predictive control}, 
  year={2015},
  volume={},
  number={},
  pages={3352-3357},
  doi={10.1109/ECC.2015.7331052}}

@article{ELALAMI1998525,
title = {On the discrete linear quadratic minimum-time problem},
journal = {Journal of the Franklin Institute},
volume = {335},
number = {3},
pages = {525-532},
year = {1998},
issn = {0016-0032},
doi = {https://doi.org/10.1016/S0016-0032(96)00133-0},
url = {https://www.sciencedirect.com/science/article/pii/S0016003296001330},
author = {N. {El Alami} and A. Ouansafi and N. Znaidi},
}

@article{Abdelhak01052010,
author = {A. Abdelhak and M. Rachik},
title = {The linear quadratic minimum-time problem for a class of discrete systems},
journal = {Optimization},
volume = {59},
number = {4},
pages = {575--587},
year = {2010},
publisher = {Taylor \& Francis},
doi = {10.1080/02331930801954672},
URL = { 
        https://doi.org/10.1080/02331930801954672  
},
eprint = { 
        https://doi.org/10.1080/02331930801954672
}
}

@article{bellman1966dynamic,
  title={Dynamic programming},
  author={Bellman, Richard},
  journal={science},
  volume={153},
  number={3731},
  pages={34--37},
  year={1966},
  publisher={American Association for the Advancement of Science}
}

@book{kirk2004optimal,
  title={Optimal control theory: an introduction},
  author={Kirk, Donald E},
  year={2004},
  publisher={Courier Corporation}
}

@INPROCEEDINGS{ipddp,
  author={Prabhu, Siddharth and Rangarajan, Srinivas and Kothare, Mayuresh},
  booktitle={2025 American Control Conference (ACC)}, 
  title={Differential Dynamic Programming with Stagewise Equality and Inequality Constraints Using Interior Point Method}, 
  year={2025},
  volume={},
  number={},
  pages={2255-2261},
  doi={10.23919/ACC63710.2025.11108083}}

@book{woodbury1950inverting,
  title={Inverting modified matrices},
  author={Woodbury, Max A},
  year={1950},
  publisher={Department of Statistics, Princeton University}
}

@INPROCEEDINGS{sympletic,
  author={Frank, Jason and Zhuk, Sergiy},
  booktitle={53rd IEEE Conference on Decision and Control}, 
  title={Symplectic Möbius integrators for LQ optimal control problems}, 
  year={2014},
  volume={},
  number={},
  pages={6377-6382},
  doi={10.1109/CDC.2014.7040389}}

@INPROCEEDINGS{sun2023receding,
  author={Sun, Jintao and Cantoni, Michael},
  booktitle={2023 62nd IEEE Conference on Decision and Control (CDC)}, 
  title={On Receding-Horizon Approximation in Time-Varying Optimal Control}, 
  year={2023},
  volume={},
  number={},
  pages={7265-7270},
  doi={10.1109/CDC49753.2023.10383405}}

@INPROCEEDINGS{gonzalez2023adaptive,
  author={Gonzalez, Carlos and Bang, Seung Hyeon and Li, Po-han and Chinchali, Sandeep and Sentis, Luis},
  booktitle={2023 62nd IEEE Conference on Decision and Control (CDC)}, 
  title={Learning Adaptive Horizon Maps Based on Error Forecast for Model Predictive Control}, 
  year={2023},
  volume={},
  number={},
  pages={8501-8508},
  doi={10.1109/CDC49753.2023.10384131}}

@conference{sima2010structure,
author={Vasile Sima},
title={STRUCTURE-PRESERVING ALGORITHMS FOR DISCRETE-TIME ALGEBRAIC MATRIX RICCATI EQUATIONS},
booktitle={Proceedings of the 7th International Conference on Informatics in Control, Automation and Robotics - Volume 1: ICINCO},
year={2010},
pages={187-192},
publisher={SciTePress},
organization={INSTICC},
doi={10.5220/0003000101870192},
isbn={978-989-8425-00-3},
issn={2184-2809},
}

\end{document}